\documentclass[11pt]{Mytran-l}
\usepackage[ansinew]{inputenc}
\usepackage[centertags]{amsmath}
\usepackage{exscale}
\usepackage{relsize}
\usepackage{amsfonts}
\usepackage{amssymb}
\usepackage{amsthm}
\usepackage{newlfont}
\usepackage{color}
\usepackage[colorlinks,citecolor=blue,urlcolor=blue,breaklinks=true]{hyperref}
\usepackage{tikz}
\usepackage{calc}

\mathchardef\ordinarycolon\mathcode`\:
  \mathcode`\:=\string"8000
  \begingroup \catcode`\:=\active
    \gdef:{\mathrel{\mathop\ordinarycolon}}
  \endgroup

\newtheorem*{mthm*}{Main Theorem}

\newtheorem{thm}{Theorem}
\newtheorem{lem}[thm]{Lemma}

\newtheorem{rem}[thm]{Remark}

\begin{document}

\title[Reconstruction and isomorphism testing]
{Computational complexity of\\ reconstruction and isomorphism testing for designs and line graphs}

\author{Michael Huber}

%

\subjclass[2000]{51E10, 05B05, 68R10, 68Q25}

\keywords{Computational complexity, reconstructibility, isomorphism testing, combinatorial design, line graph, graph isomorphism problem, hypergraph isomorphism problem}

\thanks{The author gratefully acknowledges support by the Deutsche Forschungsgemeinschaft (DFG) via a Heisenberg grant (Hu954/4) and a Heinz Maier-Leibnitz Prize grant (Hu954/5).}

\date{August 11, 2009; and in revised form June 21, 2010}

\maketitle

\vspace*{-0.5cm}

\begin{center}
{\footnotesize Wilhelm-Schickard-Institute for Computer Science\\
University of Tuebingen\\
Sand~13, D-72076 Tuebingen, Germany\\
E-mail: {\texttt{michael.huber@uni-tuebingen.de}}}
\end{center}

\smallskip

\begin{abstract}
Graphs with high symmetry or regularity are the main source for experimentally hard instances of the notoriously difficult graph isomorphism problem. In this paper, we study the computational complexity of isomorphism testing for line graphs of \mbox{$t$-$(v,k,\lambda)$} designs. For this class of highly regular graphs, we obtain a worst-case running time of $O(v^{\log v + O(1)})$ for bounded parameters $t,k,\lambda$.
In a first step, our approach makes use of the Babai--Luks algorithm to compute canonical forms of \mbox{$t$-designs}. In a second step, we show that \mbox{$t$-designs} can be reconstructed from their line graphs in polynomial-time. The first is algebraic in nature, the second purely combinatorial. For both, profound structural knowledge in design theory is required. Our results extend earlier complexity results about isomorphism testing of graphs generated from Steiner triple systems and block designs.
\end{abstract}


\section{Introduction}\label{intro}

The Graph Isomorphism (GI) problem consists in deciding whether two given finite graphs are isomorphic -- that is,
whether there exists an edge-preserving bijection between the vertex sets of the graphs. Besides of its practical importance, the inability to directly classify the GI problem into either of the conventional complexity classes $\mathsf{P}$ or $\mathsf{NP}$-complete until now have made it one of the central topics in structural complexity theory. Consequently, it is of interest to identify the difficult instances of the problem.

The best worst-case algorithm for arbitrary graphs with $v$ vertices has running time $\exp{\big(O(\sqrt{v \log v})\big)}$, see~\cite{BKL1983,BL1983}. This has mainly been achieved by a combination of Luks' seminal polynomial-time algorithm for graphs of bounded degree~\cite{Luks1982}, together with a combinatorial degree reduction due to Zemlyachenko et al.~\cite{Zem1985}. After a quarter-century, this moderately exponential bound for graph isomorphism still remains the state of the art despite extensive efforts.

Apparently, many graphs that seem to capture much of the computational difficulty are obtained from highly regular combinatorial structures, like combinatorial designs and related configurations, see~\cite{Cor1970,Math1978}. Hence, it is a primary goal to reduce for these types of graphs the leading $\sqrt{v}$ term in the exponent to $v^{1/2-\epsilon}$ for some constant $\epsilon>0$.
For important special cases, that of strongly regular graphs and that of line graphs derived from Steiner \mbox{$2$-designs},
Spielman~\cite{Spiel1996} reduced the exponent of the exponent to $1/3$ and $1/4$, respectively. For the former, Babai~\cite{Bab1980} had initially given an elementary combinatorial algorithm in $v^{O (\sqrt{v} \log v)}$ time.
Far more efficient isomorphism tests (polynomial-time or even better) are known for several parameterized classes with bounded values for their parameters. The most prominent classes are planar graphs, graphs of bounded degree, bounded genus, bounded color class, or bounded eigenvalue multiplicity. For a unifying treatment of these parameterized classes, see~\cite{Fuer1995}. A strict generalization of the results for bounded degree and bounded genus was obtained in~\cite{Mill1983b,Mill1983}.
On the other hand, GI-completeness (i.e. there exists a polynomial-time Turing reduction from the GI problem) has been proved for a number of restricted graph classes, including regular graphs, bipartite graphs, chordal graphs, self-complementary graphs, split graphs, and perfect graphs (cf.~\cite{Zem1985} for some further classes).

In this paper, we consider the computational problem of testing isomorphism of line graphs derived from \mbox{$t$-$(v,k,\lambda)$} designs. For bounded parameters $t,k,\lambda$, we obtain a sub-exponential algorithm for this important special class of the GI problem.
This extends earlier complexity results about isomorphism testing of graphs generated from Steiner triple systems and block designs.
Moreover, as \mbox{$t$-$(v,k,\lambda)$} designs can be viewed as \mbox{$k$-uniform} hypergraphs on $v$ vertices, this problem is also interesting in view of the recent moderately exponential bound for hypergraph isomorphism: Babai and Codenotti~\cite{Babai2008} have shown that isomorphism of hypergraphs of bounded rank with $v$ vertices can be tested in time $\exp{\big(\widetilde{O}(\sqrt{v})\big)}$ (where, as usual, the $\widetilde{O}$-notation suppresses polylogarithmic factors).

\smallskip

We state our main result:

\begin{mthm*}
Isomorphism of line graphs of \mbox{$t$-$(v,k,\lambda)$} designs can be determined in $O(v^{\log v + O(1)})$ time for bounded parameters $t,k,\lambda$.
\end{mthm*}

In a first step, our approach makes use of the Babai--Luks algorithm to compute canonical forms of \mbox{$t$-designs}. In a second step, we show that \mbox{$t$-designs} can be reconstructed from their line graphs in polynomial-time. The first is algebraic in nature, the second purely combinatorial. For both, profound structural knowledge in design theory is required. Specifically, we make use of the Ray-Chauduri--Wilson theorem on the minimal number of blocks, an extension of the Erd{\H{o}}s--Ko--Rado theorem to \mbox{$t$-designs} due to Rand, as well as a recent result of Kreher and Rees concerning the maximal size of a subdesign in a \mbox{$t$-design}.

\subsubsection*{\bf Related Work}

There are only a few known complexity results about isomorphism problems related to combinatorial \mbox{$t$-designs}:
Prior to Spielman's result for Steiner \mbox{$2$-designs},  Miller~\cite{Mil1978} had shown that the specific case of isomorphism of line graphs derived from Steiner triple systems (i.e. Steiner \mbox{$2$-designs} with block size $3$) can be determined in sub-exponential, $O (v^{\log v + O(1)})$, time. His proof uses the fact that a Steiner triple system can be represented as a quasigroup, and hence has a set of at most $1 + \log v$ generators. He also obtained the same bound for testing isomorphism of graphs from Latin squares. Moreover, he gave an $O (v^{\log \log v + O(1)})$ isomorphism algorithm for affine and projective planes. Miller's algorithm has been applied by M.~Colbourn~\cite{Colb1979} to perform isomorphism of Steiner \mbox{$t$-designs} with block size $t+1$ in $O (v^{\log v + O(1)})$ time.
Concerning isomorphism testing of block designs (i.e. \mbox{$2$-designs} with arbitrary $\lambda$), Babai and Luks~\cite{BL1983} derived as a consequence of Luks' techniques~\cite{Luks1982} an algorithm for bounded block size $k$ and bounded $\lambda$ in time $O (v^{\log v + f(k,\lambda)})$. On the other hand, C.~Colbourn and M.~Colbourn~\cite{Colb1981} verified that the isomorphism problem for block designs is GI-complete, even for triple systems.
For a few other results regarding specific designs, we refer to the survey~\cite[Sect.\,3]{Colb1985}. We note that the complexity of the Steiner \mbox{$t$-design} isomorphism problem in relation to the GI problem is still unresolved (even for fixed $t$). This is also the case for the isomorphism problem of Steiner triple and quadruple systems, respectively.

\subsubsection*{\bf Overview}

Relevant definitions and concepts from combinatorial design theory including line graphs will be summarized in Section~\ref{Prelim}. The reader may want to skim this section and return to it when necessary. In Section~\ref{isotest}, we apply the Babai--Luks algorithm to compute canonical forms of \mbox{$t$-designs}. In Section~\ref{Stgraph}, we show that \mbox{$t$-designs} can be reconstructed from their line graphs in polynomial-time. We finally combine the results of these sections to prove our main theorem.

\smallskip

For further detailed discussion in particular on the GI problem, we refer to the excellent literature: the books by Hoffmann~\cite{Hoff1982}, K\"{o}bler, Sch\"{o}ning and Tor\'{a}n~\cite{Koeb1993} as well as the surveys by Arvind and Tor\'{a}n~\cite{Tor2005}, Babai~\cite{Babai1995}, Booth and Colbourn~\cite{BoCol1979}, Goldberg~\cite{Gold2004}, K\"{o}bler~\cite{Koeb2006}, Read and Corneil~\cite{ReCo1977}, and Zemlyachenko et al.~\cite{Zem1985}. The current standard reference on the complexity of group-theoretic computation is Seress~\cite{Ser2003}.

\newpage


\section{Designs and Line Graphs}\label{Prelim}

\subsubsection*{\bf Combinatorial Designs}

Combinatorial design theory is a rich subject on the
interface of several disciplines, including coding and information
theory, cryptography, combinatorics, group theory, and geometry. In
particular, the study of designs with high symmetry properties has a
very long history and establishes deep connections between these
areas (see,~e.g.,~\cite{Colb1989,cosl98,Hu_isit2009,Hu2008,Hu2010,Hu_cod2008,wisl77}).

For positive integers $t \leq k
\leq v$ and $\lambda$, we define a \emph{$t$-$(v,k,\lambda)$
design} to be a finite incidence structure
\mbox{$\mathcal{D}=(X,\mathcal{B},I)$}, where $X$ denotes a set of
\emph{points}, $\left| X \right| =v$, and $\mathcal{B}$ a set of
\emph{blocks}, $\left| \mathcal{B} \right| =b$, satisfying the
following regularity properties: each block $B \in \mathcal{B}$ is
incident with $k$ points, and each \mbox{$t$-subset} of $X$ is
incident with $\lambda$ blocks. A \emph{flag} of $\mathcal{D}$ is an
incident point-block pair $(x,B) \in I$ with $x \in X$ and $B \in
\mathcal{B}$. If $t<k<v$ holds, then we speak of a \emph{non-trivial} \mbox{$t$-design}.
In this paper, `repeated blocks' are not allowed, that is, the same \mbox{$k$-element} subset
of points may not occur twice as a block. Thus, alternatively a \mbox{$t$-$(v,k,\lambda)$ design} can be viewed as a
\mbox{$k$-uniform} hypergraph on $v$ vertices with the property that every set of $t$ vertices is contained in $\lambda$ common edges.

Incidence preserving maps which take points to points and blocks to
blocks are of fundamental importance. We recall the formal definition of an isomorphism between incidence structures:
Let \mbox{$\mathcal{S}_1=(X_1, \mathcal{B}_1, I_1)$} and \mbox{$\mathcal{S}_2=(X_2,
\mathcal{B}_2, I_2)$} be two incidence structures. A bijective map
\[\alpha: X_1 \cup \mathcal{B}_1 \longrightarrow X_2 \cup \mathcal{B}_2\] is an
\emph{isomorphism}\index{isomorphism} of $\mathcal{S}_1$ onto
$\mathcal{S}_2$, if the following holds:
\begin{enumerate}
\item[(i)] for $x \in X_1$ and $B \in \mathcal{B}_1$, we have $x^{\alpha} \in
X_2$ and $B^{\alpha} \in \mathcal{B}_2$,

\smallskip

\item[(ii)] for all $x \in X_1$ and all $B \in \mathcal{B}_1$, we have
\[(x,B) \in I_1 \Longleftrightarrow (x^{\alpha}, B^{\alpha}) \in I_2.\]
\end{enumerate}
In this case, the incidence structures $\mathcal{S}_1$ and $\mathcal{S}_2$ are
\emph{isomorphic}. An isomorphism of an incidence structure $\mathcal{S}$ onto
itself is called an \emph{automorphism} of $\mathcal{S}$.
The full group of automorphisms of an incidence structure $\mathcal{S}$ will be denoted
by $\mbox{Aut}(\mathcal{S})$.

For historical reasons, a \mbox{$t$-$(v,k,\lambda)$ design} with
$\lambda =1$ is called a \emph{Steiner \mbox{$t$-design}} (sometimes
also a \emph{Steiner system}). The special case of a Steiner design with parameters $t=2$ and $k=3$
is called a \emph{Steiner triple system} $\mbox{STS}(v)$ of order $v$. A Steiner design
with parameters $t=3$ and $k=4$ is called a \emph{Steiner quadruple system} $\mbox{SQS}(v)$ of order $v$.
For example if we consider Steiner quadruple systems, the vector space $\mathbb{Z}_2^d$ with the set $\mathcal{B}$ of blocks taken to
be the set of all subsets of four distinct elements of $\mathbb{Z}_2^d$ whose vector sum is zero, is a boolean $\mbox{SQS}(2^d)$. More geometrically, these $\mbox{SQS}(2^d)$ consist of the points and planes of the \mbox{$d$-dimensional} binary affine space $AG(d,2)$.
\begin{figure}[htp]

\centering
\begin{tikzpicture}[scale=1.0,thick]

\filldraw [draw=black!100,fill=black!100]

(0,0) circle (2pt)

(1.4,0) circle (2pt)

(0.7,0.35) circle (2pt)

(2.1,0.35) circle (2pt);

\draw

(0,0) -- (0,1.4)

(0,0) -- (0.7,0.35)

(1.4,0) -- (1.4,1.4)

(1.4,0) -- (2.1,0.35)

(0.7,0.35) -- (2.1,0.35)

(2.1,0.35) -- (2.1,1.7)

(0,1.4) -- (1.4,1.4)

(0.7,1.7) -- (2.1,1.7)

(1.4,1.4) -- (2.1,1.7)

(0,1.4) -- (0.7,1.7)

(0.7,0.35) -- (0.7,1.7);

\filldraw [draw=black!100,fill=black!0]

(0,1.4) circle(2pt)

(1.4,1.4) circle (2pt)

(0.7,1.7) circle(2pt)

(2.1,1.7) circle (2pt);

\draw[very thick]

(0,0) -- (1.4,0)

(0,0) -- (0.7,0.35)

(1.4,0) -- (2.1,0.35)

(0.7,0.35) -- (2.1,0.35);

\filldraw [draw=black!100,fill=black!100]

(3.5,0) circle (2pt)

(4.9,0) circle (2pt)

(4.2,1.7) circle(2pt)

(5.6,1.7) circle (2pt);

\draw (3.5,0) -- (4.9,0)

(3.5,0) -- (3.5,1.4)

(3.5,0) -- (4.2,0.35)

(4.9,0) -- (4.9,1.4)

(4.9,0) -- (5.6,0.35)

(4.2,0.35) -- (5.6,0.35)

(5.6,0.35) -- (5.6,1.7)

(3.5,1.4) -- (4.9,1.4)

(4.2,1.7) -- (5.6,1.7)

(4.9,1.4) -- (5.6,1.7)

(3.5,1.4) -- (4.2,1.7)

(4.2,0.35) -- (4.2,1.7);

\filldraw [draw=black!100,fill=black!0]

(4.2,0.35) circle (2pt)

(5.6,0.35) circle (2pt)

(3.5,1.4) circle(2pt)

(4.9,1.4) circle (2pt);

\draw[very thick]

(3.5,0) -- (4.9,0)

(3.5,0) -- (4.2,1.7)

(4.9,0) -- (5.6,1.7)

(4.2,1.7) -- (5.6,1.7);

\filldraw [draw=black!100,fill=black!100]

(7,0) circle (2pt)

(7.7,1.7) circle(2pt)

(8.4,1.4) circle (2pt)

(9.1,0.35) circle (2pt);

\draw (7,0) -- (8.4,0)

(7,0) -- (7,1.4)

(7,0) -- (7.7,0.35)

(8.4,0) -- (8.4,1.4)

(8.4,0) -- (9.1,0.35)

(7.7,0.35) -- (9.1,0.35)

(9.1,0.35) -- (9.1,1.7)

(7,1.4) -- (8.4,1.4)

(7.7,1.7) -- (9.1,1.7)

(8.4,1.4) -- (9.1,1.7)

(7,1.4) -- (7.7,1.7)

(7.7,0.35) -- (7.7,1.7);

\filldraw [draw=black!100,fill=black!00]

(8.4,0) circle (2pt)

(7.7,0.35) circle (2pt)

(7,1.4) circle(2pt)

(9.1,1.7) circle (2pt);

\draw[very thick]

(7,0) -- (7.7,1.7)

(7,0) -- (8.4,1.4)

(7,0) -- (9.1,0.35)

(7.7,1.7) -- (9.1,0.35)

(8.4,1.4) -- (9.1,0.35)

(8.4,1.4) -- (7.7,1.7);
\end{tikzpicture}

\caption{\small{Illustration of the unique $\mbox{SQS}(8)$, with three types of blocks: faces, opposite edges, and inscribed regular tetrahedra.}}\label{cube}
\end{figure}
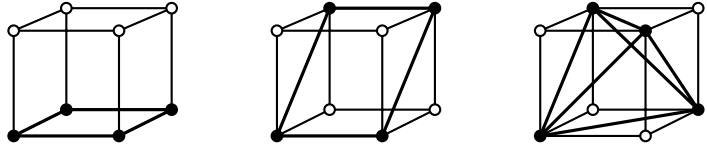
By a well-known result of Hanani, a necessary and sufficient condition for the existence of a $\mbox{SQS}(v)$ is that
$v \equiv 2$ or $4$ (mod $6$). For $v=8$ and $v=10$ there exists a \mbox{SQS$(v)$} in each case, unique up to isomorphism. These are
the affine space $AG(3,2)$ (cf.~Figure~\ref{cube}) and the M\"{o}bius plane of order $3$.
For $v=14$ there are exactly 4, and for $v=16$ exactly $1{,}054{,}163$ distinct isomorphism types.
Lenz~\cite{Lenz1985} proved that for admissible values $v$, the number $N(v)$ of non-isomorphic $\mbox{SQS}(v)$ grows exponentially, i.e.
\[\liminf_{v \rightarrow \infty} \frac{\log N(v)}{v^3}>0.\]

For a detailed treatment of combinatorial designs, we refer the reader to the encyclopedic accounts~\cite{BJL1999,crc06}.

\smallskip

We provide some combinatorial tools which will be helpful for the remainder of the paper.
For the existence of \mbox{$t$-designs}, the following basic necessary conditions can be obtained via elementary counting arguments (see,
for instance,~\cite{BJL1999}):

\begin{lem}\label{s-design}
Let $\mathcal{D}=(X,\mathcal{B},I)$ be a \mbox{$t$-$(v,k,\lambda)$}
design, and for a positive integer $s \leq t$, let $S \subseteq X$
with $\left|S\right|=s$. Then the number of blocks incident
with each element of $S$ is given by
\[\lambda_s = \lambda \frac{{v-s \choose t-s}}{{k-s \choose t-s}}.\]
In particular, for $t\geq 2$, a \mbox{$t$-$(v,k,\lambda)$} design is
also an \mbox{$s$-$(v,k,\lambda_s)$} design.
\end{lem}
\noindent It is customary to set $r:= \lambda_1$ denoting the
number of blocks incident with a given point.

\begin{lem}\label{Comb_t=5}
Let $\mathcal{D}=(X,\mathcal{B},I)$ be a \mbox{$t$-$(v,k,\lambda)$}
design. Then the following holds:
\begin{enumerate}

\item[{(a)}] $bk = vr.$

\smallskip

\item[{(b)}] $\displaystyle{{v \choose t} \lambda = b {k \choose t}.}$

\smallskip

\item[{(c)}] $r(k-1)=\lambda_2(v-1)$ for $t \geq 2$.

\end{enumerate}
\end{lem}

\begin{lem}\label{divCond}
Let $\mathcal{D}=(X,\mathcal{B},I)$ be a \mbox{$t$-$(v,k,\lambda)$}
design. Then
\[\lambda {v-s \choose t-s} \equiv \, 0\; \emph{\bigg(mod}\;\, {k-s \choose t-s}\bigg)\]
for each positive integer $s \leq t$.
\end{lem}

A generalized version of \emph{Fisher's Inequality} for \mbox{$t$-designs} by Ray-Chaudhuri and Wilson~\cite[Thm.\,1]{Ray-ChWil1975}
gives lower bounds on the number of blocks:

\begin{thm}{\em (Ray-Chaudhuri and Wilson, 1975).}\label{RayCh}
Let $\mathcal{D}=(X,\mathcal{B},I)$ be a \mbox{$t$-$(v,k,\lambda)$}
design. If $t$ is even, say $t=2s$, and $v \geq k+s$, then $b \geq
{v\choose s}$. If $t$ is odd, say $t = 2s+1$, and $v-1 \geq k+s$,
then $b \geq 2{v-1\choose s}$.
\end{thm}

\subsubsection*{\bf Line Graphs}

For an incidence structure \mbox{$\mathcal{S}=(X,\mathcal{B},I)$}, the \emph{line graph} $G(\mathcal{S})$ of $\mathcal{S}$
has as set of vertices the set $\mathcal{B}$ of blocks, whereas any two vertices are adjacent if and only if their corresponding blocks are incident with at least one common point.
Line graphs of incidence structures are sometimes alternatively called \emph{block graphs} or \emph{block intersection graphs} (or \emph{Steiner graphs} in the case of Steiner \mbox{$t$-designs}).
As an example, we consider a Steiner \mbox{$2$-$(7,3,1)$ design}, the well-known \emph{Fano plane}, which is the smallest design arising from a finite
projective geometry. Since any two of its seven blocks have a point in common, its line graph is isomorphic to the complete graph $K_7$ (see Figure~\ref{fano}). We note that a line graph of a Steiner \mbox{$2$-design} is a \emph{strongly regular graph}, i.e. each pair of adjacent vertices has the same number of common neighbors, and each pair of non-adjacent vertices has the same number of common neighbors.

\newcounter{step}
\newcommand{\completeGraph}[1]{
  \setcounter{step}{360/#1}
   \begin{tikzpicture}[scale=0.68,thick,rotate=(\thestep-174)/2,inner sep = 0pt,outer sep = 0pt]  \hspace*{-2.3cm}
    \draw
    (0,0) node (0) {};
    \tikzstyle{every node}=[draw,shape=circle,inner sep=1.3pt,fill=black!100]
    \draw
    \foreach \r in {1,...,#1}
    {
      (\r*\thestep:2) node (\r) {} -- (\r*\thestep:2)
    };
    \foreach \i in {2,...,#1}
    {
      \foreach \j in {1,...,\i}
      {
        \draw
        (\i) -- (\j);
      };
    };
  \end{tikzpicture}
}

\begin{figure}[htp]

\centering

\hskip 2cm

\begin{flushleft}

\begin{tikzpicture}[scale=1.0,thick] \hspace*{2.5cm}

\filldraw [draw=black!100,fill=black!0]

(0,0.87) circle (24.4pt);

\filldraw [draw=black!100,fill=black!100]

(0,0) circle (2pt)  (0,-0.3) node {}

(-1.5,0) circle (2pt) (-1.52,-0.3) node {}

(1.5,0) circle (2pt) (1.52,-0.3) node {}

(0,0.87) circle (2pt) (0.12,1.15) node {}

(0,2.6) circle(2pt) (0,2.9) node {}

(0.75,1.31) circle (2pt) (0.95,1.35) node {}

(-0.75,1.31) circle(2pt) (-0.95,1.35) node {};

\draw (1.5,0) -- (-1.5,0)

(1.5,0) -- (0,2.6)

(-1.5,0) -- (0,2.6)

(-1.5,0) -- (0.75,1.31)

(1.5,0) -- (-0.75,1.31)

(0,0) -- (0,2.6);

\hspace*{-3cm}

\end{tikzpicture}

\end{flushleft}

\vspace*{-3.3cm}

\vskip -2cm

\begin{flushright}

\completeGraph{7}

\end{flushright}

\caption{\small{The Fano plane $PG(2,2)$, and its line graph $K_7$.}}\label{fano}
\end{figure}
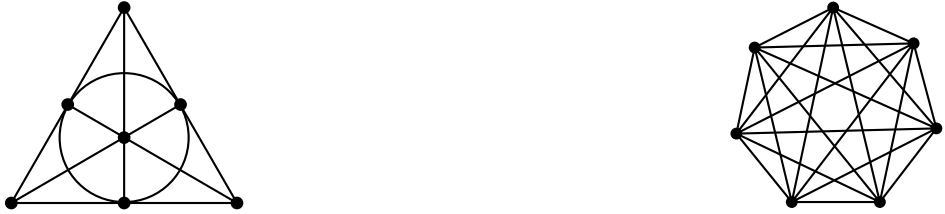

%
%
%
%
%
%
%
%
%
%
%
%
%
%
%
%
%
%
%
%
%
%
%
%
%
%
%
%
%
%
%
%
%
%

\subsubsection*{\bf Some Further Notation}

An incidence structure $\mathcal{S}_1=(X_1,\mathcal{B}_1,I_1)$ is called a \emph{substructure} of an incidence structure $\mathcal{S}=(X,\mathcal{B},I)$, if the following holds:
\begin{enumerate}
\item[(i)] $X_1 \subseteq X$ and $\mathcal{B}_1 \subseteq \mathcal{B}$,

\smallskip

\item[(ii)] for all $x \in X_1$ and all $B \in \mathcal{B}_1$, we have
\[(x,B) \in I_1  \Longleftrightarrow (x,B) \in I.\]
\end{enumerate}
A \emph{subdesign} of a \mbox{$t$-$(v,k,\lambda)$} design is a substructure of the incidence structure which itself is a
\mbox{$t$-$(w,k,\lambda)$} design. The subdesign is \emph{proper} if $w < v$.

A \emph{composition series} for a finite group $G$ is a chain of normal subgroups of the form
\[1=G^m \lhd \cdots \lhd G^2 \lhd G^1 \lhd G^0=G,\]
in which the quotients $G^i / G^{i+1}$ are simple groups. The factor groups are the \emph{composition factors} of $G$. They are independent of the choice of composition series by the Jordan--H\"{o}lder theorem. The \emph{composition width} of $G$, denoted by $\mbox{cw}(G)$, is defined to be the smallest positive integer $n$ such that every non-Abelian composition factor of $G$ embeds in the symmetric group $S_n$.

Throughout this paper, logarithms are taken base $2$. All other notation is standard.


\section{Isomorphism Testing of Designs}\label{isotest}

A standard algorithmic approach for testing isomorphism of graphs is to try to assign to each graph a \emph{canonical label} (\emph{canonical form}), so that two graphs are isomorphic if and only if the have the same label. For instance, one could start out by labeling the vertices by their degrees, and then \emph{refine} this labeling by further distinguishing equal labels through other local properties of the vertices. If, after refinement, it is possible to endow a unique label to every vertex, then a canonical label for the graph has been found. This procedure with its numerous variations has provided good algorithms for a variety of special classes of graphs. On the other hand, obstacles may occur if the graphs exhibit a high degree of regularity or symmetry, e.g. for regular graphs or graphs associated with highly regular combinatorial structures. In some cases it is possible to break up the symmetry by \emph{individualizing} particular vertices before endowing them with unique labels. For further details on the different methods used for canonical labeling, we refer to~\cite{BL1983,ReCo1977,Weis1976} and~\cite[Sect.\,2]{CaiFuer1992}.

\smallskip

Particularly important for our purposes, Miller~\cite{Mil1978} showed that a canonical labeling can be found in $O (v^{\log v + O(1)})$ time for Steiner triple systems. His proof relies on the fact that a Steiner triple system can be represented as a quasigroup, and hence has a set of at most $1 + \log v$ generators.
By individualizing these, it is then possible to order in polynomial-time the remaining vertices in a canonical way.
Babai and Luks~\cite{BL1983} extended this approach by an algebraization of the problem which involves information about the groups of automorphisms. Applied to \mbox{$2$-designs}, they obtained the subsequent result.

\begin{thm}{\em (Babai and Luks, 1983).}\label{CF_2des}
Canonical forms (and hence isomorphism testing) for non-trivial \mbox{$2$-$(v,k,\lambda)$} designs can be computed in $O(v^{\log v + f(k,\lambda)})$ time. In particular, the time bound is $O(v^{\log v + O(1)})$ for bounded parameters $k,\lambda$.
\end{thm}

A crucial observation in the Babai--Luks approach is the following well-known fact (see, e.g.,~\cite[Ch.\,II.1]{crc06}):
If there is a \mbox{$2$-$(v,k,\lambda)$} design containing a proper \mbox{$2$-$(w,k,\lambda)$} subdesign, then $v \geq (k-1)w + 1$.
As the set of all subdesigns is closed under intersection, any subset `generates' a subdesign.
In order to extend Theorem~\ref{CF_2des} to \mbox{$t$-designs}, we need a recent result by Kreher and Rees~\cite{KR2001}.

\begin{thm}{\em (Kreher and Rees, 2001).}\label{kreher}
Suppose $\mathcal{D}$ is a non-trivial \mbox{$t$-$(v,k,\lambda)$} design with $t \geq 2$ containing a proper \mbox{$t$-$(w,k,\lambda)$} subdesign. Then $v \geq 2w$ when $t$ is odd, while $v \geq 2w+1$ when $t$ is even.
\end{thm}

We can now prove the following result.

\begin{thm}\label{CF}
Canonical forms (and hence isomorphism testing) for non-trivial \mbox{$t$-$(v,k,\lambda)$} designs with $t \geq 2$ can be computed in $O(v^{\log v + f(t,k,\lambda)})$ time. In particular, the time bound is $O(v^{\log v + O(1)})$ for bounded parameters $t,k,\lambda$.
\end{thm}

\begin{proof}
Let $\mathcal{D}=(X,\mathcal{B},I)$ be a non-trivial \mbox{$t$-$(v,k,\lambda)$} design with $t \geq 2$. In view of Theorem~\ref{kreher}, we establish the key observation
\begin{enumerate}
\item[(1)] $\mathcal{D}$ has a generating set $S$ of size at most $1+ \log v$.
\end{enumerate}
By individualizing $S$, we may proceed for the remainder of the proof by straightforwardly adapting the method of proof used for Theorem~\ref{CF_2des} (cf.~\cite[Thm.\,4.6]{BL1983}). We note that this method relies on results of Luks~\cite{Luks1982}. In what follows, we describe the basic steps. We first obtain
\begin{enumerate}
\item[(2)] For fixed $t$, the composition factors of the setwise stabilizer $\mbox{Aut}_S(\mathcal{D})$ are subgroups of $S_n$, where $n=\max(\lambda,k-t)$. In particular, the composition width $\mbox{cw}(\mbox{Aut}_S(\mathcal{D}))$ is at most $n$.
\end{enumerate}
This is then employed in an inductive procedure for finding canonical forms through a nested sequence of graphs. We indicate the underlying construction for the nested graphs. For a sequence $S=(u_1, \ldots, u_s)$, a chain $\{Y_i\}_i$ of subsets of $X$ is constructed as follows: $Y_1=\{u_1\}$ and while
$Y_i \neq X$, if $Y_i$ induces a subdesign then $Y_{i+1}=Y_i \, \cup \, \{\mbox{first } u_j \mbox{ not in } Y_i\}$ else $Y_{i+1}=Y_i \, \cup \, \{B \in \mathcal{B} : \left| B \cap Y_i \right| \geq t\}$. The nested graphs $\{H_j\}_j$ are defined as bipartite graphs, $H_{2i-1}$ and $H_{2i}$, both having the set $Y_i$ on one side and on the other side the vertices representing those blocks entirely in $Y_i$ (for $H_{2i-1}$) or those in $Y_{i+1}$ (for $H_{2i}$), and edges correspond to flags.
The procedure invokes as a subroutine an algorithm of Babai and Luks (described in detail in~\cite[Sect.\,4.2]{BL1983}) for finding canonical forms for a bipartite graph with respect to a group action on one of its sides, the complexity of which is sensitive to the maximum degree on that side and to the composition width of the group.
With respect to the given construction of the nested sequence, it can be shown (again via applying techniques of Luks~\cite{Luks1982}) that the maximum degree on the side of group action is bounded by $k-t$. We therefore obtain
\begin{enumerate}
\item[(3)] For fixed $t$, the total running time is $O(v^{\log v + \omega(\max(\lambda,k-t))+O(1)})$.
\end{enumerate}
This establishes the claim.
\end{proof}


\section{Reconstruction of Designs from Line Graphs}\label{Stgraph}

If we now give an efficient method of reconstructing a \mbox{$t$-design} from its line graph, then isomorphism of line graphs of \mbox{$t$-$(v,k,\lambda)$} designs can be tested in $O(v^{\log v + O(1)})$ time for bounded $t,k,\lambda$.
To accomplish this task, we utilize an extension of the well-known Erd{\H{o}}s--Ko--Rado theorem to \mbox{$t$-designs}, which has been obtained by Rands~\cite{Rands1982}.

\begin{thm}\label{rands}
Let $\mathcal{D}=(X,\mathcal{B},I)$ be a \mbox{$t$-$(v,k,\lambda)$} design.
Given $0<s<t \leq k$, then there exists a function $f(k,t,s)$ with the following property: suppose there is a subset $\mathcal{A} \subseteq \mathcal{B}$ of blocks such that $\left| A \cap B \right| \geq s$ for all $A,B \in \mathcal{A}$, then if $v \geq f(k,t,s)$, it follows that
\[\left| \mathcal{A} \right| \leq \lambda_s \; \mbox{(with $\lambda_s$ as in Lemma~\ref{s-design})},\]
and the only families of blocks reaching this bound are those consisting of all blocks incident with an \mbox{$s$-subset} of $X$.

\smallskip

Furthermore, the function $f$ can be estimated as follows:
\[f(k,t,s) \leq \left\{\begin{array}{ll}
    s + {\displaystyle{k \choose s}}(k-s+1)(k-s) &\mbox{if} \;\, s < t-1\\
    s+ (k-s) {\displaystyle{k \choose s}^2} &\mbox{if} \;\, s= t-1.\\
\end{array} \right.\]
\end{thm}

This result will enable us to efficiently find the maximum cliques in a line graph and hence to reconstruct the points of the corresponding \mbox{$t$-design}. The idea of distinguishing \emph{cliques} (i.e. sets of mutually adjacent vertices) by simple degree considerations, and using the maximum cliques in reconstruction goes back to Miller~\cite{Mil1978}, while retrieving Latin squares, \mbox{$k$-nets}, and $\mbox{STS}(v)$. It has further been applied by Spielman~\cite{Spiel1996} in case of Steiner \mbox{$2$-designs}, and by \"{O}sterg{\aa}rd et al.~\cite{Ost2004,Ost2008} for $\mbox{STS}(v)$, $\mbox{SQS}(v)$, and Steiner \mbox{$t$-designs} via Rands' theorem.

\smallskip

We obtain the following result:

\begin{thm}\label{thm_recon}
Let $G$ be a line graph on $b$ vertices derived from a \mbox{$t$-$(v,k,\lambda)$} design $\mathcal{D}$, where $t\geq 2$.
If $b > k^2(k-1)$, then $\mathcal{D}$ can be reconstructed (up to isomorphism) in time polynomial in $b$.
\end{thm}

\begin{proof}
Let $\mathcal{D}=(X,\mathcal{B},I)$ be a \mbox{$t$-$(v,k,\lambda)$} design with $t \geq 2$. Any point $x \in X$ is incident with $r$ distinct blocks. When we consider the line graph $G(\mathcal{D})$ of $\mathcal{D}$, these blocks correspond to vertices in $G(\mathcal{D})$, and $x$ induces edges between all mutual pairs of them. Hence, the blocks intersecting in $x$ define a clique of size $r$ in $G(\mathcal{D})$. Choosing the case $s=1$ in Theorem~\ref{rands}, only this type of clique is of maximum size, if we presume that $v \geq f(k,t,1)$. Clearly, for $t \geq 2$, always $f(k,t,1) \leq 1+k^2(k-1)$, as well as $b \geq v$ by Theorem~\ref{RayCh}. Thus, under the assumption that $b > k^2(k-1)$, we may distinguish algorithmically the maximum cliques and identify them with the points of $\mathcal{D}$ in polynomial time in $b$. The claim follows.
\end{proof}

We note that $b=\Theta(v^{O(1)})$ for bounded parameters $t,k,\lambda$ in view of Lemma~\ref{Comb_t=5}~(b).

\begin{rem}\em
Spielman~\cite[Prop.\,10]{Spiel1996} elementary derived the stronger necessary condition $\sqrt{b}-2 > (k-1)^2$ in the special case of Steiner \mbox{$2$-designs}. We also remark that, in general, reconstructibility from line graphs fails for arbitrary incidence structures. The most natural and oldest graph representation of an incidence structure arguably is by its \emph{point-block incidence graph} (or \emph{Levi graph}). However, this graph representation is normally less compact.
\end{rem}

\noindent \textbf{Proof of the Main Theorem}: The result is obtained by putting together Theorem~\ref{CF} and Theorem~\ref{thm_recon}.


\subsection*{Acknowledgment}
I thank Peter Hauck, Michael Kaufmann and Jacobo Tor\'{a}n for interesting discussions about graph isomorphism, and for reading an early draft of this paper. I am also grateful for insightful suggestions from one of the anonymous referees that helped improving the presentation of the paper.

\bibliographystyle{amsplain}
\bibliography{XbibIsoTest}

\providecommand{\bysame}{\leavevmode\hbox to3em{\hrulefill}\thinspace}
\providecommand{\MR}{\relax\ifhmode\unskip\space\fi MR }
\providecommand{\MRhref}[2]{%
  \href{http://www.ams.org/mathscinet-getitem?mr=#1}{#2}
}
\providecommand{\href}[2]{#2}
\begin{thebibliography}{10}

\bibitem{Tor2005}
V.~Arvind and J.~Tor\'{a}n, \emph{Isomorphism testing: {P}erspective and open
  problems}, Bulletin of the {EATCS} \textbf{86} (2005), 66--84.

\bibitem{Bab1980}
L.~Babai, \emph{On the complexity of canonical labeling of strongly regular
  graphs}, SIAM J. Comput. \textbf{9} (1980), 212--216.

\bibitem{Babai1995}
\bysame, \emph{Automorphism groups, isomorphism, reconstruction}, in: Handbook
  of Combinatorics, ed. by R. L. Graham, M. Gr\"{o}tschel and L. Lov\'{a}sz,
  Vol. II, {North-Holland}, {Amsterdam, New York, Oxford}, 1995, 1447--1540.

\bibitem{Babai2008}
L.~Babai and P.~Codenotti, \emph{Isomorhism of hypergraphs of low rank in
  moderately exponential time}, in: Proc. 49th Annual IEEE Symposium on
  Foundations of Computer Science (Philadelphia, PA, 2008), 667--676.

\bibitem{BKL1983}
L.~Babai, W.~M. Kantor, and E.~M. Luks, \emph{Computational complexity and the
  classification of finite simple groups}, in: Proc. 24th Annual IEEE Symposium
  on Foundations of Computer Science (Tucson, AZ, 1983), 162--171.

\bibitem{BL1983}
L.~Babai and E.~M. Luks, \emph{Canonical labeling of graphs}, in: Proc. 15th
  Annual ACM Symposium on the Theory of Computing (Boston, MA, 1983), 171--183.

\bibitem{BJL1999}
Th. Beth, D.~Jungnickel, and H.~Lenz, \emph{Design {Theory}}, Vol. {I} and
  {II}, Encyclopedia of Math. and Its Applications {\bf 69/78}, Cambridge Univ.
  Press, Cambridge, 1999.

\bibitem{BoCol1979}
K.~S. Booth and C.~J. Colbourn, \emph{Problems polynomially equivalent to graph
  isomorphism}, Technical Report CS-77-04, {University of Waterloo}, 1979.

\bibitem{CaiFuer1992}
{J.-Y.} Cai, M.~F\"{u}rer, and N.~Immerman, \emph{An optimal lower bound on the
  number of variables for graph identification}, Combinatorica \textbf{12}
  (1992), 389--410.

\bibitem{Colb1981}
C.~J. Colbourn and M.~J. Colbourn, \emph{Concerning the complexity of deciding
  isomorphism of block designs}, Discrete Appl. Math. \textbf{3} (1981),
  155--162.

\bibitem{crc06}
C.~J. Colbourn and J.~H. Dinitz (eds.), \emph{Handbook of {Combinatorial}
  {Designs}}, 2nd ed., CRC Press, Boca Raton, 2006.

\bibitem{Colb1989}
C.~J. Colbourn and P.~C. van Oorschot, \emph{Applications of combinatorial
  designs in computer science}, ACM Comput. Surv. \textbf{21} (1989), 223--250.

\bibitem{Colb1979}
M.~J. Colbourn, \emph{An analysis technique for {Steiner} triple systems}, in:
  Proc. 10th Southeastern Conference on Combinatorics, Graph Theory and
  Computing (Boca Raton, FL, 1979), 289--303.

\bibitem{Colb1985}
\bysame, \emph{Algorithmic aspects of combinatorial designs: a survey}, in:
  Algorithms in Combinatorial Design Theory, ed. by C. J. Colbourn and M. J.
  Colbourn, {Annals of Discrete Mathematics \bf{26}}, {North-Holland,
  Amsterdam, New York, Oxford}, 1985, 67--136.

\bibitem{cosl98}
J.~H. Conway and N.~J.~A. Sloane, \emph{Sphere {P}ackings, {L}attices and
  {G}roups}, 3rd ed., Springer, Berlin, Heidelberg, New York, 1998.

\bibitem{Cor1970}
D.~G. Corneil and C.~C. Gotlieb, \emph{An efficient algorithm for graph
  isomorphisms}, J. ACM \textbf{17} (1970), 51--64.

\bibitem{Fuer1995}
M.~F\"{u}rer, \emph{Graph isomorphism testing without numerics for graphs of
  bounded eigenvalue multiplicity}, in: Proc. 6th Annual ACM-SIAM Symposium on
  Discrete Algorithms (San Francisco, CA, 1995), 624--631.

\bibitem{Gold2004}
M.~Goldberg, \emph{The graph isomorphism problem}, in: Handbook of Graph
  Theory, ed. by J. L. Gross and J. Yellen, {CRC Press}, {Boca Raton}, 2004,
  68--78.

\bibitem{Hoff1982}
C.~M. Hoffmann, \emph{{Group-Theoretic Algorithms} and {Graph Isomorphism}},
  Lecture Notes in Comp. Science {\bf 136}, Springer, Berlin, Heidelberg, New
  York, 1982.

\bibitem{Hu_isit2009}
M.~Huber, \emph{Authentication and secrecy codes for equiprobable source
  probability distributions}, in: Proc. IEEE International Symposium on
  Information Theory (Seoul, South Korea, 2009), 1105--1109.

\bibitem{Hu2008}
\bysame, \emph{Flag-transitive {S}teiner {D}esigns}, {Frontiers in Mathematics,
  Birkh\"{a}user}, Basel, Berlin, Boston, 2009.

\bibitem{Hu2010}
\bysame, \emph{Combinatorial {D}esigns for {A}uthentication and {S}ecrecy
  {C}odes}, {Foundations and Trends$^\circledR$ in Communications and
  Information Theory, Now Publishers}, Boston, Delft, 2010.

\bibitem{Hu_cod2008}
\bysame, \emph{Coding theory and algebraic combinatorics}, in: Selected Topics
  in Information and Coding Theory, ed. by I. Woungang et al., {World
  Scientific}, {Singapore}, 2010, 121--158.

\bibitem{Ost2004}
P.~Kaski and P.~R.~J. \"{O}sterg{\aa}rd, \emph{The {Steiner} triple systems of
  order $19$}, Math. Comput. \textbf{73} (2004), 2075--2092.

\bibitem{Koeb2006}
J.~K\"{o}bler, \emph{On graph isomorphism for restricted graph classes}, in:
  Proc. 2nd Conference on Computability in Europe, Logical Approaches to
  Computational Barriers, ed. by A.~Beckmann et al., Lecture Notes in Comp.
  Science \textbf{3988}, Springer, {Berlin, Heidelberg, New York}, 2006,
  241--256.

\bibitem{Koeb1993}
J.~K\"{o}bler, U.~Sch\"{o}ning, and J.~Tor\'{a}n, \emph{{The Graph Isomorphism
  Problem: Its Structural Complexity}}, Birkh\"{a}user, Basel, Berlin, Boston,
  1993.

\bibitem{KR2001}
D.~L. Kreher and R.~S. Rees, \emph{A hole-size bound for incomplete $t$-wise
  balanced designs}, J. Combin. Designs \textbf{9} (2001), 269--284.

\bibitem{Lenz1985}
H.~Lenz, \emph{On the number of {Steiner} quadruple systems}, Mitt. Math. Sem.
  Giessen \textbf{169} (1985), 55--71.

\bibitem{Luks1982}
E.~M. Luks, \emph{Isomorphism of graphs of bounded valence can be tested in
  polynomial time}, J. Comput. System Sci. \textbf{25} (1982), 42--65.

\bibitem{wisl77}
F.~J. MacWilliams and N.~J.~A. Sloane, \emph{The {Theory} of {Error-Correcting}
  {Codes}}, North-Holland, Amsterdam, New York, Oxford, 1977; 12. impression
  2006.

\bibitem{Math1978}
R.~A. Mathon, \emph{Sample graphs for isomorphism testing}, in: Proc. 9th
  Southeastern Conference on Combinatorics, Graph Theory and Computing (Boca
  Raton, FL, 1978), 499--517.

\bibitem{Mil1978}
G.~L. Miller, \emph{On the $n^{\log n}$ isomorphism technique: A preliminary
  report}, in: Proc. 10th Annual ACM Symposium on the Theory of Computing (San
  Diego, CA, 1978), 51--58.

\bibitem{Mill1983b}
G.~L. Miller, \emph{Isomorphism of graphs which are pairwise $k$-separable},
  Information and Control \textbf{56} (1983), 21--33.

\bibitem{Mill1983}
\bysame, \emph{Isomorphism of $k$-contractible graphs. {A} generalization of
  bounded valence and bounded genus}, Information and Control \textbf{56}
  (1983), 1--20.

\bibitem{Ost2008}
I.~Yu. Mogilnykh, P.~R.~J. \"{O}sterg{\aa}rd, O.~Pottonen, and F.~I. Solov'eva,
  \emph{Reconstructing extended perfect binary one-error-correcting codes from
  their minimum distance graphs}, IEEE Trans. Inform. Theory \textbf{55}
  (2009), 2622–2625.

\bibitem{Rands1982}
B.~M.~I. Rands, \emph{An extension of the {E}rd{\H{o}}s, {K}o, {R}ado theorem
  to $t$-designs}, J. Combin. Theory, Series A \textbf{32} (1982), 391--395.

\bibitem{Ray-ChWil1975}
D.~K. Ray-Chaudhuri and R.~M. Wilson, \emph{On $t$-designs}, Osaka J. Math.
  \textbf{12} (1975), 737--744.

\bibitem{ReCo1977}
R.~C. Read and D.~G. Corneil, \emph{The graph isomorphism disease}, J. Graph
  Theory \textbf{1} (1977), 339--363.

\bibitem{Ser2003}
A.~Seress, \emph{Permutation {G}roup {A}lgorithms}, {Cambridge Univ. Press},
  Cambridge, 2003.

\bibitem{Spiel1996}
D.~A. Spielman, \emph{Faster isomorphism testing of strongly regular graphs},
  in: Proc. 28th Annual ACM Symposium on the Theory of Computing (Philadelphia,
  PA, 1996), 576--584.

\bibitem{Weis1976}
B.~Weisfeiler (ed.), \emph{On {C}onstruction and {I}dentification of {G}raphs},
  Lecture Notes in Math. \textbf{558}, Springer, Berlin, Heidelberg, New York,
  1976.

\bibitem{Zem1985}
V.~M. Zemlyachenko, N.~M. Kornienko, and R.~I. Tyshkevich, \emph{Graph
  isomorphism problem}, J. Soviet Math. \textbf{29} (1985), 1426--1481.

\end{thebibliography}
\end{document}